\documentclass{llncs}

\usepackage{hyperref}
\usepackage{graphicx}
\usepackage{float}
\usepackage{comment}
\usepackage{subfig}
\usepackage{footnote}
\usepackage{complexity}

\DeclareGraphicsExtensions{.pdf}
\graphicspath{{./img/}}

\begin{document}
\title{PSPACE-completeness of Bloxorz
\\and of Games with 2-Buttons}

\author{Tom C. van der Zanden \and Hans L. Bodlaender}

\institute{Department of Computer Science, Utrecht University \\ \{\email{tom@tomvanderzanden.nl}, \email{H.L.Bodlaender@uu.nl}\}}

\newlength{\problemoffset}
\setlength{\problemoffset}{0in}

\newcommand{\decision}[3]{
\begin{list}{}{
\setlength{\leftmargin}{\problemoffset}
\setlength{\rightmargin}{\problemoffset}
\setlength{\parsep}{0pt}
\setlength{\itemsep}{2pt}
\setlength{\topsep}{\itemsep}
\setlength{\partopsep}{\itemsep}
}
\item
{\textsc{#1}}
\item
{\textbf{Instance:} #2}
\item
{\textbf{Question:} #3}
\end{list}
}

\maketitle

\thispagestyle{plain} 

\begin{abstract}
Bloxorz is an online puzzle game where players move a 1 by 1 by 2 block by tilting it on a subset of the two dimensional grid. Bloxorz features switches that open and close trapdoors. The puzzle is to move the block from its initial position to an upright position on the destination square. We show that the problem of deciding whether a given Bloxorz level is solvable is \PSPACE-complete and that this remains so even when all trapdoors are initially closed or all trapdoors are initially open. We also answer an open question of Viglietta \cite{vig14}, showing that 2-buttons are sufficient for \PSPACE-hardness of general puzzle games. We also examine the hardness of some variants of Bloxorz, including variants where the block is a 1 by 1 by 1 cube, and variants with single-use tiles.
\end{abstract}

\section{Introduction}
We study the computational complexity of the online puzzle game Bloxorz. We show that deciding whether it is possible to reach the goal square in a level of Bloxorz featuring switches and trapdoors is \PSPACE-complete, which we will prove by reduction from \textsc{Nondeterministic Constraint Logic}. We first give a proof that requires us to be able to choose the initial state (open or closed) of each trapdoor. We will then also show that the problem remains \PSPACE-complete even when the trapdoors are required to be initially all closed or all open.

Viglietta \cite{vig14} studied the complexity of general computer games featuring certain common elements such as ``destroyable paths, collectible items, doors opened by keys or activated by buttons or pressure plates, etc''. Viglietta established that a game featuring 2-buttons and doors is \NP-hard and that a game featuring 3-buttons is \PSPACE-hard. Viglietta asked whether this result could be improved to show that a game with 2-buttons is \PSPACE-hard. We settle this question, showing that 2-buttons are indeed sufficient for \PSPACE-hardness.

We also examine several variants of Bloxorz, and compare the hardness of decision and optimization versions of these. We show that Bloxorz with single-use tiles but without switches or trapdoors is \NP-complete. Bloxorz with a 1 by 1 by 1 cube instead of the 1 by 1 by 2 block is \NP-complete when both single-use tiles and switches and trapdoors are included, but becomes polynomially solvable if either is not present in the level.

\subsection{Bloxorz}
Bloxorz \cite{dx} is a ``brain twisting puzzle game of rolling blocks and switching bridges'' by Damien Clarke. A 1 by 1 by 2 block is tilted around a subset of the two dimensional grid. Games using similar mechanics to Bloxorz are available on multiple platforms, including Android \cite{caveman1} and Apple iOS \cite{caveman2}.

\begin{figure}[t]
\centering
\includegraphics[scale=0.5]{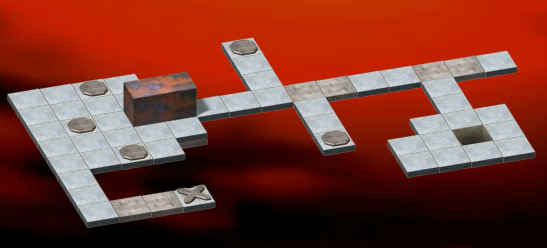}
\caption{An example Bloxorz level.}
\label{fig:level}
\end{figure}

The block can be in two states: lying down on a rectangular side, or standing up on a square face. From a standing position, the player may make a tilting move (Figure \ref{fig:move1}) to place the block on its side. From a lying position, the player can either make a tilting move (Figure \ref{fig:move1}) to stand the block up, or a rolling move (Figure \ref{fig:move2}) after which the block will still be in a lying state.

\begin{figure}[h]
     \centering
     \hfill
     \subfloat[][Tilting move] {
         \includegraphics[height=3cm]{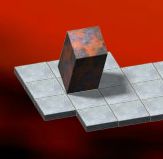}
         \label{fig:move1}
     }
     \hfill
     \subfloat[][Rolling move] {
         \includegraphics[height=3cm]{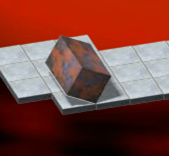}
         \label{fig:move2}
     }
     \hfill\null
     \caption{The two types of move available: (a) tilting and (b) rolling}
\end{figure}

Not all moves are possible: the player may only make a move if after that move the block is still fully supported by the game level. In Bloxorz, if the player does attempt to make such a move the block will topple off the game stage and the player will have to restart the level.

The goal of the puzzle is to reach some specified goal square (Figure \ref{fig:goal}) and get the block to fall through it. This implies the block must be in a standing orientation upon reaching the goal square but this constraint does not contribute to the hardness.

The online version of Bloxorz features various gadgets intended to make the game more challenging: switches that open and close trapdoors, switches that can only be triggered with the block in a specific orientations, weak tiles that will not support the block standing up, block-splitting portals and more. We will consider the version of Bloxorz with only switches, trapdoors and of course, regular squares. An example of a trapdoor and switch is shown in Figure \ref{fig:bridge}.

A \emph{trapdoor} is a special game square that can be either open or closed. If the trapdoor is open, it is not allowed for the block to be on top of it. Each trapdoor is controlled by exactly one switch square and each switch only controls one trapdoor.

A \emph{switch} is a special game square that is associated with exactly one trapdoor. If a move places the block over a switch, the state of the associated trapdoor is toggled (from open to closed and vice-versa). A switch may be triggered more than once. If a move places a block over 2 switches then they will both be triggered.

\begin{figure}[t]
     \centering
     \hfill
     \subfloat[][Trapdoor and switch] {
         \includegraphics[height=3cm]{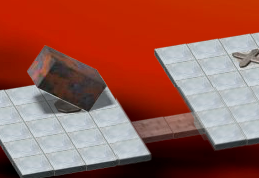}
         \label{fig:bridge}
     }
     \hfill
     \subfloat[][Goal square] {
         \includegraphics[height=3cm]{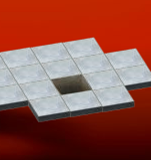}
         \label{fig:goal}
     }
     \hfill
     \subfloat[][Reaching the goal square] {
         \includegraphics[height=3cm]{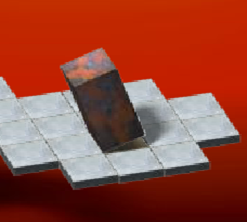}
     }
     \caption{More game elements: (a) switches and bridges, (b),(c) goal square}
\end{figure}

If a move places a block over a switch and its associated trapdoor simultaneously, the effect of pressing the switch will be considered before checking the legality of the move.  Our results remain valid if we check the legality of the move before activating the switch.

Note that in the original version of the game trapdoors are instead called ``bridges'' and are of size $1\times2$ but our constructions will use $1\times1$ trapdoors. However, our proofs can easily be adapted to work with $1\times2$ trapdoors instead.

We now formally define \textsc{Bloxorz} as decision problem:

\decision{Bloxorz}{A level of Bloxorz, given as a list of game tiles (normal tiles, trapdoors and switches) with their positions, a one-to-one mapping between trapdoors and switches, for each trapdoor an initial state (open/closed) and start/goal squares.}{Is there a sequence of legal moves from the start to the goal square?}

The precise encoding used for problem instances is not important. We consider the size of a level to be the number of tiles in it.

\subsection{Previous work}
Buchin \cite{buchin12} showed that solving rolling block mazes is \PSPACE-complete. The block in Bloxorz moves in exactly the same way as a block in a rolling block maze, but a rolling block maze may feature multiple blocks (which may interfere with each other's movement) while Bloxorz only features a single block. Bloxorz also differs from rolling block mazes in that it features switches and trapdoors. These are required for the hardness since rolling block mazes with only a single block are solvable in polynomial time.

Viglietta \cite{vig14} established several metatheorems regarding the complexity of games with buttons\footnote{In the conference version of his paper \cite{vig12}, Viglietta uses the term ``switch'' instead of ``button''. In the final version of the paper \cite{vig14}, he uses ``button'' to describe the same object. We use this newer terminology since it better describes the functionality of the object.}, doors and pressure plates.

A \emph{door} is a special game square that may be either in a closed or in an open state, and the avatar may pass through the door if and only if it is open. Bloxorz features trapdoors which are similar to doors except that the notions of open and closed are reversed (i.e. the avatar may pass over the trapdoor if and only if it is closed).

The state of a door may be altered by pressure plates and buttons. Viglietta defines a \emph{pressure plate} as a special game square ``that is operated whenever the avatar steps on it, and its effect may be either the opening or the closure of a specific door''.

In addition, Viglietta defines \emph{buttons} which work similar to pressure plates except that the player has a choice whether to trigger the button or not, rather than a pressure plate which is immediately activated when the player steps on it.

In Bloxorz, a switch instead toggles the state of a trapdoor between open and closed. A pressure plate or button may either open or close a door, but may not toggle between the two states. Like a pressure plate, a switch is toggled when the block moves over it. Due to the $1$ by $1$ by $2$ shape of the block two switches may be triggered simultaneously.

A single door may be controlled up to to two pressure plates: one that opens it, and one that closes it. In Bloxorz, one trapdoor is controlled by exactly one switch.

A $k$-button may control up to $k$ doors simultaneously, possibly opening some while closing others. Viglietta established that a game featuring doors and $k$-buttons with $k\geq2$ is \NP-hard, while a game with $k\geq3$ is \PSPACE-hard.

Viglietta poses the question whether his result for $k=2$ can be improved upon, that is, if a game with $2$-buttons is \PSPACE-hard. We show that this is the case.

Our reduction is based on the work by Hearn and Demaine \cite{hearn02} on Nondeterministic Constraint Logic.

\section{Nondeterministic Constraint Logic}
First, we will briefly introduce (restricted) Nondeterministic Constraint Logic (NCL).

A \emph{constraint graph} is an undirected graph $G=(V,E)$ with for each vertex and edge a weight that is a positive integer.
We only need to consider weights of 1 and 2. The weight of a vertex is called its \emph{minimum inflow}.

A \emph{configuration} is an assignment of an orientation to each edge, which is \emph{legal} if and only if, for each vertex $v$,
the sum of the weights of the edges pointing into $v$ is at least the minimum inflow of $v$.

An \emph{edge reversal} is the operation that reverses the orientation of an edge. Thus, 
its weight starts counting towards the other vertex. A reversal of $(v,w)$
is \emph{legal} if the minimum inflow constraints for $v$ and $w$ remain satisfied.

We will consider only constraint graphs featuring two special types of vertices:

\begin{itemize}
\item An \emph{AND vertex} is a vertex of degree 3, with minimum inflow 2 and with incident edge weights of $1$, $1$ and $2$. It works like a logical AND gate because either the weight 2 edge must be directed inwards or the two weight 1 edges must be directed inwards.

\item An \emph{OR vertex} is a vertex of degree 3, with minimum inflow 2 and with incident edge weights all equal to 2. It works like a logical OR since to satisfy its minimum inflow requirement, at least one of its edges must be directed inwards.
\end{itemize}

\decision{Restricted NCL}{A constraint graph $G$ built exclusively from AND and OR vertices, a legal initial configuration for $G$ and a given target edge $e$ from $G$.}{Is there a sequence of legal edge reversals that, starting from the specified initial configuration, reverses $e$?}

\begin{theorem}[Hearn and Demaine \cite{hearn02}]
\textsc{Restricted NCL} is \PSPACE-complete.
\end{theorem}

Note that \textsc{Restricted NCL} remains \PSPACE-complete when restricted to planar graphs, but we do not make use of this property in our proof. This is because the structure of the NCL graph will be encoded in the switch-trapdoor correspondence in the Bloxorz level, rather than in its physical structure.

\section{PSPACE-completeness of Standard Bloxorz}

\begin{theorem}
\textsc{Bloxorz} is \PSPACE-complete.
\label{theorem:standardbloxorz}
\end{theorem}

\begin{proof}
By reduction from \textsc{Restricted NCL}. We first establish \PSPACE-hardness. We first describe the various constructions used to represent elements from a constraint graph, then show how to put them together in to a \textsc{Bloxorz} level. Finally, we show that \textsc{Bloxorz} is in \PSPACE.
\end{proof}

\subsection{Overview}

Figure \ref{fig:construction} shows a simplified overview of our construction. The white area represents a section of the game level where the block can move freely. The grey boxes represent copies of the vertex and edge gadgets, which will be described in detail later. The start and goal squares are marked with $S$ and $G$ respectively. 

\begin{figure}[th]
\centering
\includegraphics[scale=0.8]{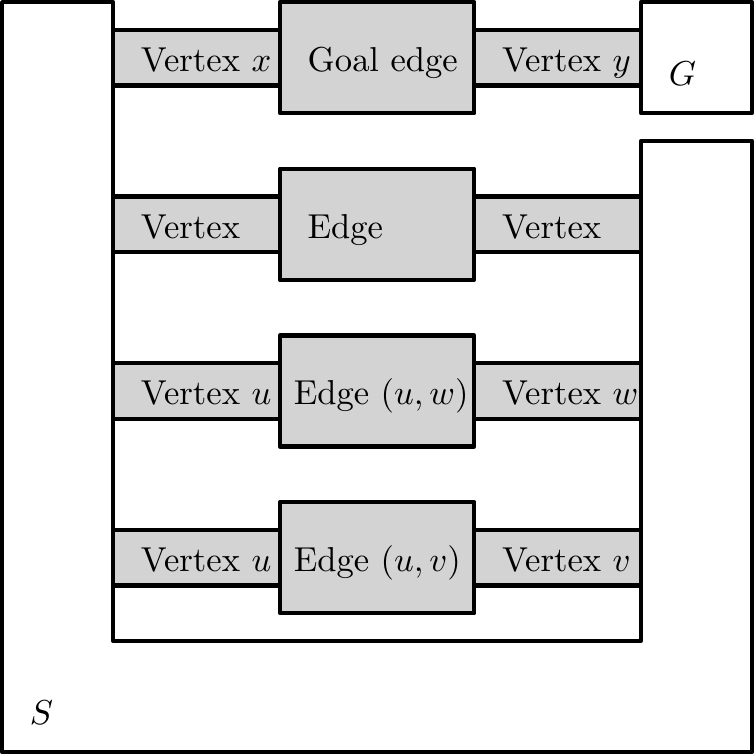}
\caption{Overview of the construction}
\label{fig:construction}
\end{figure}

A solution to the Bloxorz level corresponds to a solution to the NCL problem as follows: the player will cross the gadgets from left to right and from right to left a number of times to solve the level. Crossing the edge gadget $(u,v)$ from left to right (towards vertex gadget $v$) corresponds to changing the orientation of the edge $(u,v)$ to be towards $u$. To cross the edge gadget, the player will need to activate a number of switches in the edge gadget which changes the state of the trapdoors in the vertex gadgets.

The vertex gadgets will be designed so that it is possible to leave the edge via them only if their inflow constraint is satisfied. So, if the player crosses the edge $(u,v)$ from left to right, orienting the edge towards $u$, the vertex $v$ loses inflow while $u$ gains inflow. Since gaining inflow can not cause a constraint to be violated, we only need to check the constraint at vertex $v$. This checking is done in the vertex gadget.

Note that the input graph to the \textsc{restricted NCL} problem only has degree 3 vertices. For every vertex in the input graph, 3 copies of the vertex gadget are created - one for each of the incident edges. These are related by switch-trapdoor correspondence: in our example, crossing the edge $(u,v)$ would also affect the state of the trapdoors in the gadget for vertex $u$ attached to edge $(u,w)$.

In \textsc{NCL}, the goal is to eventually reverse a given goal edge. This is encoded in our Bloxorz level by the placement of the goal square: in order to reach it the player has to cross the edge gadget of the goal edge which corresponds to reversing its orientation.

\subsection{Edge gadget}
Figure \ref{fig:edge_gadget} shows the construction used to represent an edge from an NCL graph. The gray squares represent normal tiles, the squares with fat borders represent trapdoors. Circles represent switches, and the lettering shows switch-trapdoor correspondence. Trapdoors and switches with white infill correspond to open trapdoors and conversely, grey switches and trapdoors correspond to closed trapdoors. The switches labelled with $u,v$ correspond to trapdoors that are inside vertex gadgets, which will be introduced later.

\begin{figure}[th]
\centering
\includegraphics[scale=1]{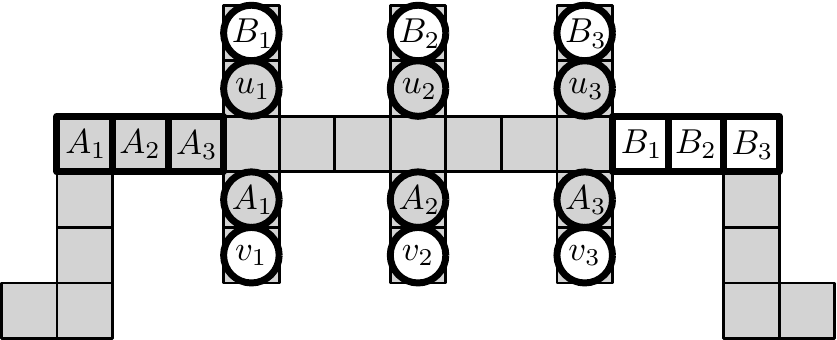}
\caption{The edge gadget}
\label{fig:edge_gadget}
\end{figure}

On either end, the edge gadget will be attached to a corner gadget. The edge gadget is shown pointing to the right vertex, i.e. its weight is counting towards the right vertex.

To reorient the edge gadget, the player would enter it from the left and travel over the trapdoors labelled $A$. They would then trigger all the switches exactly once, which leaves the trapdoors labelled $A$ open while closing the trapdoors labelled $B$. The block can then leave the edge via the right side, travelling over the trapdoors labelled $B$. Note that due to the 1 by 1 by 2 shape of the block it is only possible for the switches to be triggered in pairs (e.g. $B_1$ is always triggered at the same time as $u_1$) so this leaves the trapdoors labelled $v$ closed and the ones labelled $u$ open. The edge gadget is now in the opposite orientation (left) and to reverse it again, one would make the same moves but starting from the right.

Note that when we reorient the edge from right to left, the right vertex loses weight while the left vertex gains weight. This means that the only place a constraint could possibly have been violated is the right vertex. Since the trapdoors labelled $A$ are now open, we have to leave via the right side by travelling over the trapdoors labelled $B$. The vertex gadget will be constructed in such a way that it enforces the constraint: the vertex gadget will prevent us from leaving the edge gadget if its constraint is not satisfied.

Above, we described a ``canonical'' movement of the block over the edge. Clearly, the player can decide to move the block over the edge gadget in many different ways. For example, they could enter it from the left, activate all the switches on the top side, closing the trapdoors labelled $B$ and opening those labelled $u$. They could then leave the edge again from the left side, and the edge's weight is essentially removed from the graph: it no longer gives its weight to either vertex. And even more cause for concern: we did this without checking that the constraint of the right vertex was not violated!

However it is never beneficial to only complete a reorientation partially. If at any point we want the weight of the edge to start counting towards the left vertex we have to open (at least one of) the trapdoors labelled $A$ which forces us to exit via the right side and check the constraint for the right vertex. Also, note that the edge will never give its weight to more than one vertex since that would mean both a trapdoor labelled $A$ is open and one labelled $B$ is open which makes it impossible to leave the edge.

\begin{claim}
Upon leaving the edge gadget, it is not possible for a trapdoor labelled $u$ to be closed while a trapdoor labelled $v$ is also closed.
\end{claim}

\begin{proof}
Leaving the edge gadget requires all the trapdoors $A$ to be closed or all the trapdoors labelled $B$ to be closed, which in turn enforces the condition of the lemma. This enables us to define the orientation of the edge gadget: the edge gadget is oriented away from the vertex $u$ if the trapdoors labelled $u$ are open and oriented away from $v$ if all the trapdoors labelled $v$ are open. If both trapdoors $u$ and $v$ are open then we may consider the edge gadget oriented either way.
\qed
\end{proof}

\subsection{Vertices}

A vertex is represented by three separate, identical gadgets, that attach to the edge gadgets corresponding to its incident edges. The edge has three switches corresponding to trapdoors in each of the three copies of the vertex gadget. Every vertex has three trapdoors, and every trapdoor is associated with a distinct edge connecting to the vertex.

\subsubsection{OR vertex gadget}

Figure \ref{fig:orv} shows one of the three identical components that together form one OR vertex gadget. It is drawn with part of an edge gadget attached. The trapdoors labelled $x,y$ and $v_1$ each correspond to one of the three edges incident to the vertex.
The following claim is self-evident. Note that it enforces the minimum inflow constraint: if one of the trapdoors is closed then the edge associated with that trapdoor must be pointing towards this vertex.

\begin{claim}
It is possible for the block to leave the edge via the OR vertex gadget if and only if at least one of the trapdoors $x, y$ or $v_1$ is closed.
\end{claim}

\begin{figure}[th]
     \centering
     \hfill
     \subfloat[][OR vertex] {
         \includegraphics[scale=1]{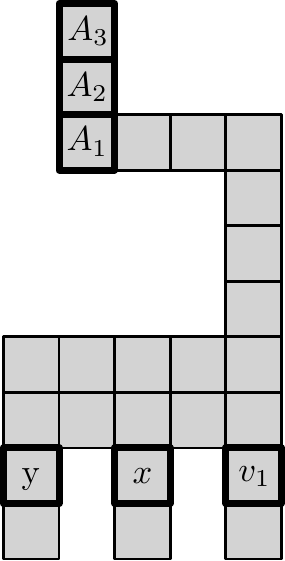}
         \label{fig:orv}
     }
     \hfill
     \subfloat[][AND vertex] {
         \includegraphics[scale=1]{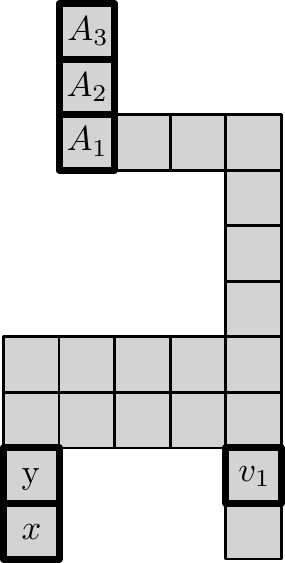}
         \label{fig:andv}
     }
     \hfill\null
     \caption{Vertex gadgets, showing the three identical components that make up: (a) an OR vertex, (b) an AND vertex. The vertex gadgets attach to edge gadgets at the top and to the rest of the game level at the bottom. Note that the trapdoors labelled with $A$ are part of the edge gadget.}
\end{figure}

\subsubsection{AND vertex gadget}

Figure \ref{fig:andv} shows one of the three identical components that together form one AND vertex. The following claim regarding the AND vertex gadget is self-evident:

\begin{claim}
It is only possible for the block to leave the edge via the AND vertex gadget if and only if the trapdoor labelled $v_1$ is closed or both the trapdoors labelled $x$ and $y$ are closed.
\end{claim}

Thus, the  edge associated with trapdoor $v_1$ functions as the weight-2 edge and the edges associated with trapdoors $x$ and $y$ function as weight-1 edges. The vertex gadget enforces that the minimum inflow constraint is satisfied.

\subsection{Final details}

We will now show how to assemble the component gadgets, forming a \textsc{Bloxorz} instance.

For every edge in the constraint graph, one copy of the edge gadget is created. For every vertex, three copies of the appropriate vertex gadget are created and attached to their respective edges. Then, we add some additional blocks so that every vertex can be reached from every other vertex, with the sole exception of the vertex in to which the target edge is pointing. it is not attached to any other vertex, instead, it is attached to the goal square. We can place the start square anywhere on the level, except on those squares from which the goal square can be reached directly.

\begin{claim}
An instance of \textsc{Bloxorz} constructed in this way from a \textsc{Restricted NCL} instance has a solution
if and only if there exists a sequence of edge reversals that reverses the target edge in the \textsc{NCL} instance.
\end{claim}

\begin{proof}
A solution for the \textsc{NCL} instance, consisting of a sequence of edge reversals, translates directly to a solution for the Bloxorz instance by traversing the edge gadgets in the same order as the edges are reversed.

A solution that solves the \textsc{Bloxorz} instance can be translated in to a solution for the \textsc{NCL} instance as follows: every time we traverse an edge gadget associated with edge $(u,v)$ and leave via the vertex gadget $u$ corresponds to reorienting edge $(u,v)$ away from $u$. Since such a reversal only increases the weight of $v$, its minimum inflow must remain satisfied. The inflow constraint for $u$ must remain satisfied since this is enforced by the vertex gadget.
\qed
\end{proof}

\begin{claim}
\textsc{Bloxorz} is in \PSPACE.
\end{claim}

\begin{proof}
Since each state can be represented in polynomial space and we can generate the successor states efficiently, 
a nondeterministic search of the state space shows that Bloxorz is in \NPSPACE.
By Savitch`s theorem\cite{sav70}, \NPSPACE = \PSPACE, so Bloxorz is in \PSPACE.
\qed
\end{proof}

Since we have shown that \textsc{Bloxorz} is both contained in $\PSPACE$ and \PSPACE-hard, \textsc{Bloxorz} is \PSPACE-complete, and thus have shown Theorem~\ref{theorem:standardbloxorz}.
\qed

Our proof requires that we are able to choose the initial state of each trapdoor (open or closed) when specifying the problem instance. The following theorems show this is not essential to the hardness:

\begin{theorem}
\textsc{Bloxorz} remains \PSPACE-complete even when limited to the instances where all trapdoors are initially open.
\end{theorem}

\begin{proof}
Figure \ref{fig:eopen} illustrates how we modify the construction to
force the player to close a specific trapdoor. In order to get from the start to the goal, the user must pass the trapdoor labelled $x$, but this forces them to also activate the switch $u_1$. We thus modify the level such that the player must first travel through such a gadget for every trapdoor that should initially be closed before they can access the rest of the level. For every edge gadget, we need to duplicate this construction six times.
\qed
\end{proof}

\begin{figure}[h]
     \centering
     \hfill
     \subfloat[][Force closing] {
         \includegraphics[scale=1]{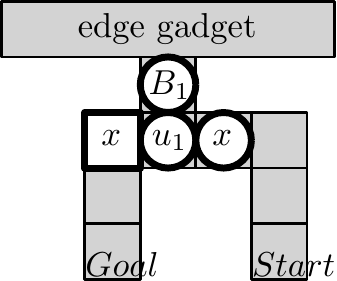}
         \label{fig:eopen}
     }
     \hfill
     \subfloat[][Force opening] {
         \includegraphics[scale=1]{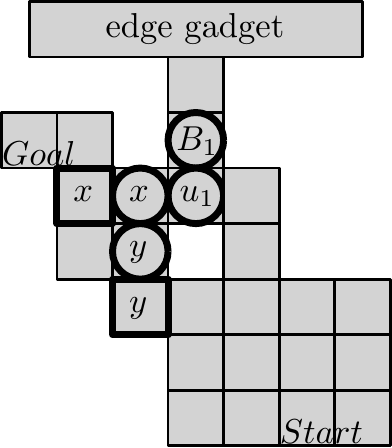}
         \label{fig:eclosed}
     }
     \hfill\null
     \caption{Constructions that show how we can: (a) force the player to close a specific trapdoor, (b) force the player to open a specific trapdoor. $B_1$ and $u_1$ correspond to switches that are part of an edge gadget.}
\end{figure}

\begin{theorem}
\textsc{Bloxorz} remains \PSPACE-complete even when limited to the instances where all trapdoors are initially closed.
\end{theorem}

\begin{proof}
Figure \ref{fig:eclosed} shows a construction that forces the player to open trapdoor $u_1$. In order to get from the start to the goal the player must travel over the switches $x$ and $y$. If the player attempts to do this with both switches closed, travelling over them will open trapdoors $x$ and $y$ and the player gets stuck.

The player must first trigger the switches $u_1$ and $x$ so that they are open. Only then can they pass over the switches $x$ and $y$ since trapdoor $x$ will be closed while $y$ will be opened.

This construction needs to be duplicated 6 times per edge and linked together so that we are forced to open all the appropriate trapdoors before reaching the rest of the level.
\qed
\end{proof}

Note that in both these cases it is possible to later, while solving the rest of the level, undo previous moves and travel back up these gadgets, but upon leaving the gadgets again the parity of the switches can not have changed.

Furthermore, in both proofs connecting the gadgets requires the path of the block to cross over itself. Crossovers are possible in Bloxorz due to the 1 by 1 by 2 shape of the block. By designing the level carefully, we can ensure the block always reaches the crossing in a specific orientation so that it can not cross on to the other path.

\section{Hardness of general puzzles with buttons}

Viglietta \cite{vig14} established that a game featuring 2-buttons and doors is \NP-hard and asked whether this could be improved to show such games are \PSPACE-hard. We show that such games are indeed \PSPACE-hard.

Note that for a door, the notions of open and closed are inverted from that of a trapdoor.

\begin{theorem}\label{theorem:buttongame}
A game where the avatar has to reach an exit location to win that features 2-buttons and doors is \PSPACE-hard.
\end{theorem}

\begin{proof}
Our construction to show \PSPACE-hardness of \textsc{Bloxorz} can be adapted to show \PSPACE-hardness of a game featuring 2-buttons and doors. 
Replace every trapdoor by a door. In the edge gadget of Figure \ref{fig:edge_gadget} replace every pair of switches $B_i$ and $u_i$ by a pair of 2-buttons: one 2-button that opens $B_i$ and closes $u_i$ and another that closes $B_i$ and opens $u_i$. Do the same for switches labelled $A$ and $v$. This new construction functions in the exact same way as the edge gadget for Bloxorz.
\qed
\end{proof}

Figure \ref{fig:viglietta} shows an example of this construction. The small white rectangles depict buttons, and are labelled to show which doors they open (marked with a plus sign) and which doors they close (marked with a minus sign). The doors labelled $x_1,\ldots,x_4$ correspond to other incident edges of $u$ and $v$.

\begin{figure}[th]
\centering
\includegraphics[scale=1.05]{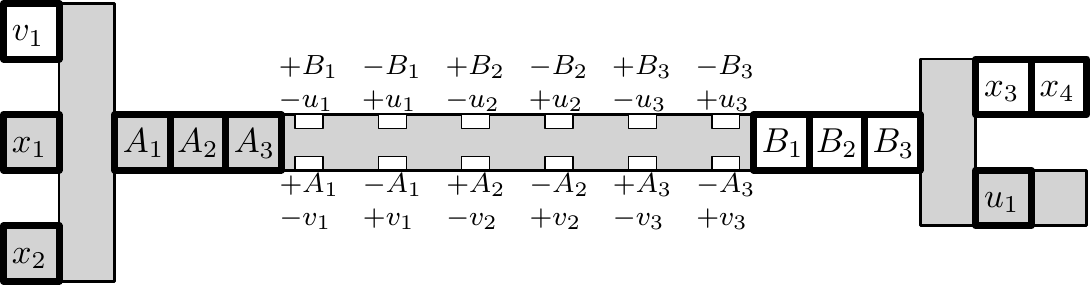}
\caption{Construction showing the edge gadget, attached to an OR vertex gadget on the left and an AND vertex gadget on the right, modified to work with 2-button games.}
\label{fig:viglietta}
\end{figure}

\begin{theorem}\label{theorem:viginitial}
A game that matches the conditions of Theorem \ref{theorem:buttongame} is \PSPACE-hard, even when all doors are initially closed. When all doors are initially open the problem can be decided in polynomial time.
\end{theorem}

\begin{proof}
We can create a construction that allows the player to open certain doors once at the beginning of a level, using additional doors and buttons to enforce that this can only be done once at the beginning of the level. The details of this construction are omitted for brevity. When all doors are initially open we do not need to trigger any buttons and can walk straight to the goal of the level if there exists a path. We can check this in polynomial time using a standard path finding algorithm.
\qed
\end{proof}

\subsection{$k$-Doors}

Similar to the notion of a $k$-button (that controls up to $k$) doors, we may also consider $k$-doors: a door that is controlled by up to $k$ buttons. Our proof for the \PSPACE-hardness of games with 2-buttons works with 2-doors. This result can not be improved upon, since if a game has no other features than 2-buttons and 2-doors, it can be decided by a $\PSPACE$ algorithm.

A game with 2-buttons and 1-doors is \NP-hard \cite{vig14}. This can not be improved upon either, as a game with no other features than 2-buttons and 1-doors can be decided by a non-deterministic polynomial algorithm that guesses which buttons and in what order they are pressed. Since pressing a button is idempotent, we do not need to press any button more than once.

A game with 1-buttons and $k$-doors can be decided in polynomial time using a path finding algorithm, pressing any buttons that open doors as it finds them. If such games feature crossovers (i.e. the level can be non-planar) it is P-hard \cite{vig14}.

In summary, we have the following results:

\begin{table}[h]
\centering
\renewcommand{\arraystretch}{1.2}
\begin{tabular}{ | c | c | c | }
  \hline			
 & $l=1$ & $l\geq 2$  \\
 \hline

  $k=1$ & \textsc{P} & \textsc{P} \\
  \hline

  $k\geq 2$ & \NP & \PSPACE \\
   \hline
 
\end{tabular}
\caption{Hardness of games with $k$-buttons and $l$-doors}
\end{table}

It is an open problem whether a game with 2-buttons and 2-doors remains \PSPACE-hard when all doors are initially closed. Theorem \ref{theorem:viginitial} that shows that games with 2-buttons are \PSPACE-hard even when all doors are initially closed requires 3-doors.

\section{Variants of Bloxorz}

\begin{savenotes}
\begin{table}[H]
\centering
\renewcommand{\arraystretch}{1.2}
\begin{tabular}{ | c || c | c | c || c | c | c | }
  \hline			
 & \multicolumn{3}{c|}{$\mathbf{1\times1\times2}$} & \multicolumn{3}{|c|}{$\mathbf{1\times1\times1}$}  \\
 \hline
   & Dec. & Opt. & Moves &  Dec. & Opt. & Moves \\
  \hline \hline
 - & \P & \P & $\Theta(n)$ 
   & \P & \P & $\Theta(n)$ \\  \hline

  Switches+Trapdoors & \PSPACE-C & \PSPACE-C & $2^{\Theta(n)}$\footnote{\label{foot2}This follows from our \PSPACE-completeness proof. Solving a configuration created in this way from an \textsc{NCL} instance can take exponentially many moves.}
  & \P & \NP-C\footnote{\label{foot1}Metatheorem 2 of \cite{fori10}. Collecting items is implemented by forcing a number of trapdoors to be closed before the exit square can be reached.} & $\Theta(n^2)$ \\ \hline
  
     Single-use tiles & \NP-C & \NP-C & $\Theta(n)$ 
    & \P & \P & $\Theta(n)$ \\  \hline
    
    S+T+Single-use & \PSPACE-C & \PSPACE-C & $2^{\Theta(n)}$\textsuperscript{\ref{foot2}}
    & \NP-C\footnote{Metatheorem 1 of \cite{vig14}. Location traversal can be implemented by forcing the player to close a number of trapdoors.} & \NP-C & $\Theta(n^2)$ \\  \hline
    
\end{tabular}
\caption{Complexity of various Bloxorz variants}
\label{tab:complexity}
\end{table}
\end{savenotes}

In order to more closely identify what features make Bloxorz hard, we now look at various simplified variants and determine their complexity. Table \ref{tab:complexity} shows the complexity for various variants: we list the complexity of the corresponding decision problem (determining if a solution exists) and optimization problem (determining whether a solution with at most $k$ moves exists) and also the number of moves required in the worst case. We will not discuss all of these results as some of them are trivial.

We consider the regular 1 by 1 by 2 variant of Bloxorz but also consider what happens when we use a 1 by 1 by 1 cube instead. We look at various game elements: trapdoors and switches as previously considered but also single-use tiles. Single-use tiles are game squares that the block can only pass over once before they break and are removed from the game board. Single-use tiles are actually not featured in the original Bloxorz game. We discuss them here since they are an example of a game feature that, when combined with switches and trapdoors will make the decision for the 1 by 1 by 1 case \NP-complete.

\begin{theorem}
1 by 1 by 2 Bloxorz with single-use tiles is NP-complete, while 1 by 1 by 1 Bloxorz with single-use tiles can be solved in polynomial time.
\end{theorem}

\begin{proof}
We show that 1 by 1 by 2 Bloxorz with single-use tiles is \NP-complete by reduction from \textsc{3-sat}.

\begin{figure}[th]
\centering
\includegraphics[scale=0.8]{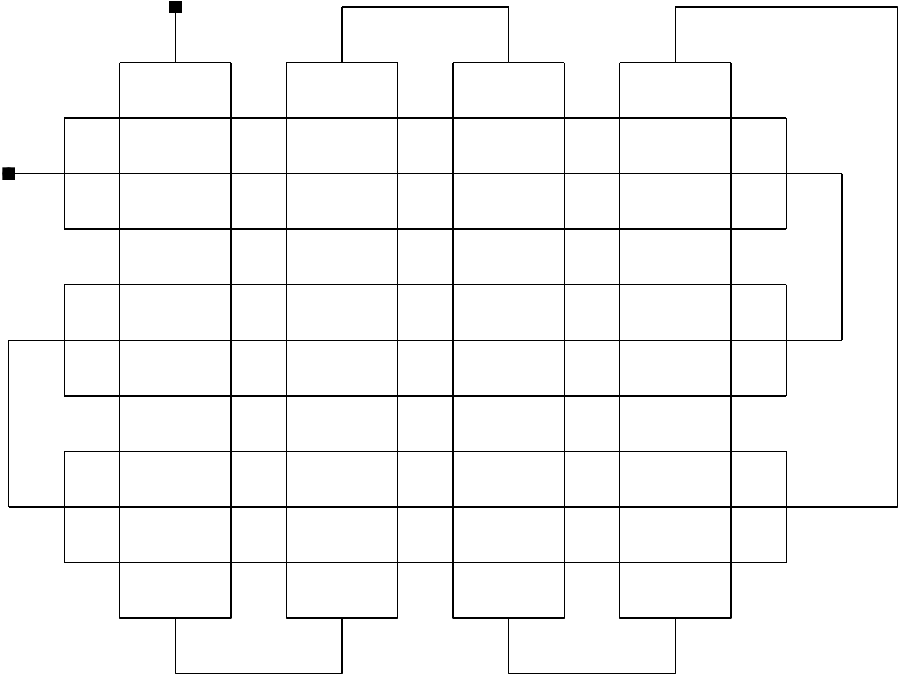}
\caption{Example for NP-hardness proof of Bloxorz with 1 by 1 by 2 block and single-use tiles}
\label{fig:bloxorz-dis}
\end{figure}

See Figure~\ref{fig:bloxorz-dis} for an example; it represents an instance
with four variables and three clauses. The two marked positions are the starting and ending squares. By properly choosing the lengths of parts of the construction, we can ensure
that the block when reaching a cross between a horizontal and vertical line is not in upright position, so the block must continue on the
same row or column and cannot `take a turn'. 
When reaching at
a branch between two vertical lines, the player must choose one of these, which corresponds to setting a variable to true or false. 
Each set of three linked horizontal lines represents a clause; the player must choose one of these to traverse. 
Thus, by construction, a path from start to target must necessarily use one from each pair of
vertical lines (these represent variables and the choice represents which literal is true)
and one from each triple of horizontal lines (these represent clauses; the player must choose the branch corresponding to a true literal).
If a horizontal line represents a literal from a clause, and a vertical line represents the negation of a literal, then the square where
these lines cross is a single-use square: in this way, we enforce that each clause has a true literal. Membership in $\NP$ is trivial. \qed
\end{proof}

We can decide Bloxorz with trapdoors and switches in polynomial time when we do not use a rectangular block but a 1 by 1 by 1 cube:

\begin{theorem}
In the 1 by 1 by 1 case, a Bloxorz instance with switches and trapdoors be solved in polynomial time, taking $\theta(n^2)$ moves in the worst case.
\end{theorem}

\begin{proof}
The construction of Figure \ref{fig:bloxorz-dis} cannot be used for the 1 by 1 by 1 cube, as this block can take turns at crossings. In the case of a 1 by 1 by 1 block, we never need to visit a square more than once, and thus we can treat single use squares as normal squares, and a simple BFS will find a shortest solution in polynomial time.

1 by 1 by 1 Bloxorz with switches and trapdoors can be solved in polynomial time using at most quadratically many moves. This is because whenever we trigger a switch opening a trapdoor we can trigger the same switch again to close it, either by taking a step back in our path or taking a step forward in our path and then undoing that same move. We can keep track of a set of tiles that we can reach without having to pass a currently open trapdoor, and continually expand this set by closing one more trapdoor until this set includes the goal square or there are no more trapdoors we can close. If a solution exists, this algorithm can find one of length $O(n^2)$.

To show that 1 by 1 by 1 Bloxorz with switches and trapdoors can require a quadratic number of moves, we can use a linear layout, with successive squares $s_r$, $t_{r-1}$, $s_{r-2}$, \ldots, $t_3$, $s_2$, $t_1$, the
starting square, $s_1$, $t_2$, $s_3$, \ldots, $s_{r-1}$, $t_r$ and the target square, with $s_i$ the switch corresponding to trapdoor $t_i$. 
It is easy to see that a solution involves the player to successively touch all the switches and thus repeatedly move from left to right
and back, making a quadratic number of moves.
\qed
\end{proof}

However, when trapdoors, switches and single-use tiles are combined, 1 by 1 by 1 Bloxorz becomes \NP-complete again.

\section{Conclusion}

We have shown by a reduction from \textsc{restricted NCL} that \textsc{Bloxorz} is \PSPACE-complete, even when the trapdoors are initially fixed either open or closed. We also showed how our proof can be adapted to general games and showed that games with 2-buttons and doors are not only \NP-hard, but also \PSPACE-hard, answering an open question of Viglietta \cite{vig14}.

We then examined some other variants of Bloxorz, including variants where the 1 by 1 by 2 block is replaced by a 1 by 1 by 1 cube and variants where single-use tiles are included. We showed that Bloxorz with single-use tiles but without switches or trapdoors is \NP-complete, and that most variants with a 1 by 1 by 1 block can be polynomially solved. Only the 1 by 1 by 1 variant featuring both single-use tiles and switches and trapdoors is \NP-complete.

We think our result clearly illustrates the power of the NCL framework. It is not at all obvious that an universal quantifier may be constructed from the elements in Bloxorz;
the work by Demaine and Hearn \cite{hearn02,hdbook} is instrumental for establishing
the complexity of these and other puzzle and reconfiguration problems.

\subsubsection*{Acknowledgement}
We thank Damien Clarke for graciously allowing us to reproduce screenshots of Bloxorz.

\bibliographystyle{splncs}

\end{document}